%% file: arxiv-main.tex
\documentclass{article}

\usepackage{arxiv}

\usepackage[utf8]{inputenc} %
\usepackage[T1]{fontenc}    %
\usepackage{url}            %
\usepackage{nicefrac}       %
\usepackage{microtype}      %
\usepackage{lipsum}		%
\usepackage{doi}

\usepackage{cite}
\usepackage{amsmath,amsthm, amssymb,amsfonts}
\usepackage{graphicx}
\usepackage{comment}
\usepackage{subcaption}
\usepackage{mwe}
\usepackage{hyperref}
\usepackage{epstopdf}
\usepackage{adjustbox}

\usepackage[utf8]{inputenc}
\usepackage{algpseudocode}

\usepackage{multirow}
\usepackage{booktabs}
\usepackage{tabularx}
\usepackage{textcomp}
\newcolumntype{x}[1]{>{\centering\arraybackslash\hspace{0pt}}p{#1}}

\usepackage[ruled,vlined,linesnumbered]{algorithm2e}
\usepackage[dvipsnames]{xcolor}

\newcolumntype{C}{>{\centering\arraybackslash}X}

\usepackage{listings} %
\usepackage{seqsplit} %
\usepackage{xcolor} %

\usepackage{array}

\newcolumntype{M}{>{\hsize=.8\hsize\raggedright\arraybackslash}X}  %
\newcolumntype{D}{>{\hsize=1.6\hsize\raggedright\arraybackslash}X} %

\newtheorem{theorem}{Theorem}

\newtheorem{corollary}[theorem]{Corollary}
\newtheorem{proposition}[theorem]{Proposition}

\usepackage[most]{tcolorbox}
\usepackage{mdframed} %
\newtcolorbox{observationbox}{
    colback=green!3!white, %
    colframe=black, %
    coltitle=black,          %
    fonttitle=\bfseries,     %
    boxrule=1pt,             %
    arc=5mm,                 %
    breakable,               %
    enhanced,                %
}

\newcommand{\validators}{\mathcal{V}}
\newcommand{\validator}[1]{v_{#1}}
\newcommand{\contract}{\mathcal{SC}}

\newif\ifratcolor
\ratcolortrue
\ifratcolor
  \colorlet{rat}{NavyBlue!90!black}
\else
  \colorlet{rat}{black}
\fi
\newcommand{\RAT}[1]{\textcolor{rat}{#1}}

\title{Looking for Attention: Randomized Attention Test Design for Validator Monitoring in Optimistic Rollups}

\author{%
  Suhyeon Lee \\
  Tokamak Network \\
  \texttt{orion-alpha [at] korea.ac.kr}
  \And
  Yeongju Bak \\
  Tokamak Network \\
  \texttt{zena [at] tokamak.network}
}

\renewcommand{\shorttitle}{Randomized Attention Test Design}

\hypersetup{
pdftitle={Looking for Attention: Randomized Attention Test Design for Validator Monitoring in Optimistic Rollups},
pdfsubject={CoRR},
pdfauthor={Suhyeon Lee,Yeongju Bak},
pdfkeywords={Ethereum, blockchain security, cryptographic protocol, decentralized application, gas optimization, verifiable delay function},
}

\begin{document}
\maketitle

\begin{abstract}
Optimistic Rollups (ORUs) significantly enhance blockchain scalability but inherently suffer from the verifier's dilemma, particularly concerning validator attentiveness. Current systems lack mechanisms to proactively ensure validators are diligently monitoring L2 state transitions, creating a vulnerability where fraudulent states could be finalized. This paper introduces the Randomized Attention Test (RAT), a novel L1-based protocol designed to probabilistically challenge validators in ORUs, thereby verifying their liveness and computational readiness. Our game-theoretic analysis demonstrates that an Ideal Security Equilibrium, where all validators are attentive and proposers are honest, can be achieved with RAT. Notably, this equilibrium is attainable and stable with relatively low economic penalties (under \$1000) for non-responsive validators, a low attention test frequency (under 1\% per epoch), and a minimal operation overhead (monthly under \$30) with 10 validators. RAT thus provides a pivotal, practical mechanism to enforce validator diligence, fortifying the overall security and integrity of ORU systems with minimizing additional costs.
\end{abstract}

\keywords{attention test \and Ethereum \and economics \and layer 2 \and mechanism design \and optimistic rollup \and security}

\input{1.intro}

\input{2.related}

\input{3.problem}

\input{4.protocol}

\input{5.game_model}

\input{6.equilibrium}

\input{7.design}

\input{8.implementation}

\input{9.discussion}

\input{10.conclusion}

\section*{Acknowledgement}

The authors thank the anonymous reviewers for their helpful comments.

\appendix

\section{Code Availability} \label{appendix:code}

The proof-of-concept implementation in Section~\ref{sec:implementation} is available at the following GitHub repository: \url{https://github.com/tokamak-network/optimism/tree//feature/rat-poc-v1}.

\bibliographystyle{ieeetr}
\bibliography{references}

\end{document}

%% file: 1.intro.tex
\section{Introduction}
\label{section:introduction}
Layer 2 scaling solutions, particularly Optimistic Rollups (ORUs), have become crucial for enhancing blockchain throughput and reducing transaction costs while inheriting the security of the underlying Layer 1 (L1). ORUs achieve scalability by processing transactions off-chain and posting only essential data and state commitments to L1. Their security relies on the ``optimistic'' assumption that state transitions posted by the proposer are valid, backed by a challenge period during which any validator can submit a Fraud Proof to dispute invalid states.

However, the fundamental security of ORUs hinges critically on the assumption that there exists at least one honest and diligent validator actively monitoring the L1, downloading L2 data, re-executing transactions, and submitting Fraud Proofs when necessary. This assumption is potentially fragile due to the well-known verifier's dilemma \cite{luu2015demystifying}. Validators may lack direct incentives to perform the costly verification work, especially if the proposer is perceived as honest or if they believe other validators will bear the cost (free-riding). This dilemma is particularly pronounced in ORUs because, during normal operation where no fraud occurs, validators have minimal, if any, on-chain interactions. This lack of observable activity makes it difficult to ascertain whether validators are genuinely online and performing their verification duties, or merely appearing to be so. Consequently, a lazy but online validator, failing to perform actual verification, creates a significant security vulnerability, as a malicious proposer could potentially finalize fraudulent states unchallenged. Existing liveness checks (e.g., network pings) are insufficient as they do not ascertain a validator's capability or willingness to perform the core verification tasks.

While the concept of an ``attention test'' to proactively verify validator engagement has been discussed within the blockchain community for several years as a potential solution \cite{felten2019attention, park2021optimistic, mamageishvili2023incentive}, concrete protocol designs and analyses of their practical implications have been notably absent. To address this critical gap, we propose RAT (Randomized Attention Test), the first comprehensive attention test mechanism specifically designed for Optimistic Rollups. RAT introduces probabilistic, individually-targeted challenges integrated into the proposer's state commitment process. When triggered, a randomly selected validator is required to independently compute the state transition and engage in a scheme with the proposer within a strict time limit. Failure to respond correctly or timely results in direct economic penalties, thus providing individual economic incentives for validators to maintain operational attentiveness.

Our main contributions are as follows:
\begin{itemize}
    \item To the best of our knowledge, we are the first to detail a practical attention challenge mechanism, termed Randomized Attention Test (RAT), specifically designed to address the lazy validator problem in Optimistic Rollups.
    \item We present a comprehensive game-theoretic analysis demonstrating that RAT effectively incentivizes rational proposers and validators to converge towards an Ideal Security Equilibrium, where all validators remain attentive and proposers act honestly even under the weak randomness environment.
    \item We present a proof-of-concept implementation using the Optimism Stack, a representative rollup codebase, demonstrating that RAT is deployable with minimal costs.
    \item Our findings indicate that this Ideal Security Equilibrium is attainable with RAT under realistic conditions, imposing only minimal operational overhead on the ORU system, primarily through modest, probabilistically applied economic penalties and low challenge frequencies.
\end{itemize}

The remainder of this paper is structured as follows. Section \ref{sec:related_works} briefly reviews the related works. Section \ref{sec:problem} defines the problem scope. Section \ref{sec:protocol} introduces the design of the RAT protocol, a novel attention test mechanism. Section \ref{sec:game_model} explains the game model of validators and proposer strategies in RAT. Based on the model, Section \ref{sec:equilibrium_analysis} analyzes the conditions and robustness of the ideal strategic equilibrium by RAT. Section \ref{sec:design_ideal_security} investigates implications of the parameter setting and practicality for RAT based on the game theoretic results. 
Section \ref{sec:implementation} provides a PoC implementation and the cost analysis.
Section \ref{sec:discussion} discusses limitations in modeling, and Section \ref{sec:conclusion} concludes the paper.

%% file: 2.related.tex
\section{Related Works}
\label{sec:related_works}

The development of optimistic rollups has spurred a wide array of research aimed at enhancing their economic viability, security, and operational scalability. This section reviews key areas of this prior work, highlighting common assumptions and identifying a nuanced challenge related to proactive validator oversight, which our present study seeks to address.

\subsection{Economic Analysis in Rollups}
Early works in the economic analysis of rollups have focused on establishing robust security guarantees through well-aligned incentive mechanisms. Tas et al.~\cite{tas2022accountable} propose a framework for \emph{accountable safety} in rollups, which emphasizes the need for designs that hold participants economically accountable for misbehavior. Li~\cite{li2023security} further explores the security of optimistic blockchain mechanisms, highlighting the importance of ensuring that validator incentives are strong enough to deter malicious actions. In a similar vein, Mamageishvili and Felten~\cite{mamageishvili2023incentive} analyzes incentive schemes for rollup validators, including a discussion of strategies under a conceptual attention test. These foundational studies establish the importance of economic incentives for honest participation and for penalizing malicious behavior. However, a subtle aspect of the verifier's dilemma persists: even with robust penalties for detected fraud, the incentive for an individual validator to incur the cost of active verification during periods of perceived sequencer honesty can be weak, especially if they assume others will bear this cost. This raises questions about the consistent, proactive engagement of validators, a prerequisite for these incentive systems to function effectively when fraud does occur.

\subsection{Data Availability and Its Cost}
Given that Optimistic Rollups are designed as a layer-2 scaling solution, minimizing the costs associated with data availability (DA) is a critical concern. Several studies have tackled this problem from different angles. Palakkal et al.~\cite{Palakkal2024sok} provide a systematization of compression techniques in rollups, outlining methods and inefficiencies in some rollups' practice. Mamageishvili and Felten~\cite{mamageishvili2023efficient} propose efficient rollup batch posting strategies on the base layer using call data, while Crapis et al.~\cite{crapis2023eip} offer an in-depth analysis of EIP-4844 economics and rollup strategies, including blob cost sharing. Complementary empirical investigations by Heimbach and Milionis~\cite{Heimbach2025}, Huang et al.~\cite{huang2024two}, and Park et al.~\cite{park2024impact} further document the inefficiencies in current DA cost structures and examine their impact on rollup transaction dynamics and consensus security. Lee~\cite{lee2024180} investigates blob sharing as a potential remedy for the dilemma faced by small rollups in the post-EIP-4844 era. These advancements are crucial for ensuring that validators \textit{can} access the necessary data affordably and efficiently. Yet, the availability and affordability of data do not inherently guarantee that every validator will consistently download, process, and verify it, especially if there are no direct, periodic checks on their actual engagement with this data beyond dispute scenarios.

\subsection{Fraud Proof and Security Protocols}
Prior research on fraud proofs in Optimistic Rollups has primarily focused on ensuring rapid dispute resolution and robust liveness of the dispute game itself. For example, BoLD~\cite{alvarez2024bold} and Dave~\cite{nehab2024dave} propose protocols that optimize the dispute game to achieve low delays and cost-efficient on-chain verification. Concurrently, Berger et al.~\cite{berger2025economic} analyze economic censorship dynamics to guarantee that disputes are resolved even under adversarial conditions. While these protocols are vital for the \textit{reactive} security of ORUs by providing mechanisms to challenge and penalize identified fraudulent state transitions, some analyses, such as Lee~\cite{lee2025wtsc}, point to potential exploitabilities in current validator incentive structures by malicious proposers.

Sheng et al.~\cite{shengPoD2024AFT} proposed a Proof-of-Diligence (PoD) protocol which provides cryptoeconomic security for rollup validators. PoD operates a watchtower network in which watchtowers recompute state transitions and, when selected by the VRF, submit a bounty (PoD) transaction. Challenges are raised if an incorrect state root is detected. It leverages a Verifiable Random Function (VRF) to pseudo-randomly select bounty recipients, embedding this into a game-theoretic incentive design. Under normal operation, each state epoch yields \emph{at most one} L1 bounty (PoD) transaction (380,000 gas). 
At recent Ethereum gas and price levels (see Section~\ref{sec:implementation}), a PoD transaction costs approximately \$2.4. When assuminng a 10-minute epoch, this can translate to a monthly on-chain cost of up to \$10k. Importantly, PoD's cost is controlled by the system parameter $\theta$. However, the original paper offers no detailed guidance on how to tune this parameter for different security needs. Using the only provided exmaple, $\theta=0.9$, the expected on-chain cost would be around \$9,000 per month.
On the other hand, our work, RAT, leverages the existing validators without deploying an additional network. When a rollup proposer posts a new state root, the contract probabilistically triggers an attention test targeting a validator. The selected validator must submit a test solution, which can be verified on-chain at low cost. This mechanism makes validator liveness externally observable and allows us to tune the trigger probability to keep the cost overhead minimal. With a 10-minute epoch and 10 validators setting, the on-chain cost is under \$23 per month. Table \ref{tab:rat-vs-pod-simple} provides a side-by-side comparison of RAT and PoD.

\begin{table}[t]
\centering
\small
\setlength{\tabcolsep}{6pt}
\begin{tabularx}{\textwidth}{lXX}
\toprule
\textbf{Metric} & \textbf{Randomized Attention Test} & \textbf{Proof-of-Diligence \cite{shengPoD2024AFT}} \\
\midrule
Infrastructure & No extra network & Watcher network required \\
Submitted artifact & Minimal hashes to rebuild posted state root & Recomputed state/trace Merkle \\
Contract-side verification & Yes (hash recompute \& compare in SC) & No (off-chain cross-check) \\
On-chain cost\footnotemark[1] & $\approx$ 23 \$/month & $\leq$ 10,000 \$/month \\
Liveness observability & Directly on-chain & Indirect, only visible upon misbehavior \\
Limitation & No guarantee of full execution & No on-chain verification \\
\bottomrule
\end{tabularx}
\caption{RAT vs.\ PoD}
\label{tab:rat-vs-pod-simple}
\end{table}

\footnotetext[1]{Monthly on-chain cost estimates use a gas price of 2.3 gwei, an ETH/USD rate of \$2,753, and a 10-minute epoch period (see Section~\ref{sec:implementation}). For RAT, we assume 10 validators and a 1\% trigger probability. For PoD, we assume one claim per epoch, with 380k gas per claim \cite{shengPoD2024AFT}.}

%% file: 3.problem.tex
\section{Problem Statement and Design Goals}
\label{sec:problem}

This section first outlines the core problems within Optimistic Rollups that necessitate the development of an attention test mechanism. Subsequently, we establish key design goals that such a mechanism should satisfy, which will serve as an evaluative framework for the solution proposed in this paper.

\subsection{The Validator Liveness and Correctness Problem}
\label{subsec:validator_problem}
The security model of Optimistic Rollups critically depends on the assumption that there exists at least one honest Validator who is actively monitoring the L1, downloading L2 data, executing state transitions, and submitting Fraud Proofs when necessary. However, several challenges arise:

\begin{itemize}
    \item \textbf{The \textbf{Lazy} Validator Problem:} A validator might remain online and appear active but fail to perform the computationally intensive task of re-executing state transitions. Such a validator would not detect invalid state roots posted by a malicious proposer, undermining the rollup's security.
    \item \textbf{Limitations of Simple Liveness Checks:} Basic liveness checks (like monitoring network connectivity) are insufficient as they do not verify if the validator possesses the necessary data or has performed the required computations.
    \item \textbf{Reactive Nature of Fraud Proofs:} The Fraud Proof mechanism is inherently reactive. It only comes into play \textbf{after} an invalid state has been posted. There is no built-in mechanism to proactively ensure that validators are continuously ready and capable of performing their verification duties \textbf{before} a potential fraud occurs.
\end{itemize}

Therefore, a mechanism is needed to proactively and periodically verify that Validators are not only online (live) but also correctly processing the L2 state (correctness) and have access to the necessary data (data availability). We term such a mechanism an \textit{Attention Test}.

\subsection{Design Goals for Attention Tests}
\label{subsec:design_goals}
An effective attention test mechanism for ORU validators should possess the following properties:

\begin{enumerate}
    \item \textbf{Verifiability:} The test should check the validator's ability to perform core validation tasks (like state computation), and the results must be objectively verifiable on L1.
    \item \textbf{Unpredictability:} The timing of the test and the selection of the challenged Validator should be unpredictable to prevent gaming the system.
    \item \textbf{Fairness:} All active validators should have a chance of being selected for a test, without bias towards specific individuals or groups.
    \item \textbf{Efficiency:} The mechanism should incur low overhead in terms of L1 gas costs and latency added to the normal rollup operation. (This is often a trade-off against security).
\end{enumerate}

The subsequent sections will introduce the RAT (Randomized Attention Test) protocol, designed to address the verifier's dilemma and validator liveness challenges outlined previously. We will then detail its operation and analyze how it fulfills the design goals established.

%% file: 4.protocol.tex
\section{The RAT Protocol}
\label{sec:protocol}

To address the challenges outlined in Section \ref{sec:problem}, specifically the verifier's dilemma and the need for proactive validator monitoring, we propose the \textbf{R}andomized \textbf{A}ttention \textbf{T}est (RAT) protocol. RAT integrates a probabilistic challenge mechanism into the routine state commitment process of the optimistic rollup, leveraging the existing structure on the L1 blockchain to minimize additional costs.

\subsection{System Model and Assumptions}
\label{subsec:system_model}

This study operates within a standard ORU framework, which comprises the following key entities and components:

\begin{itemize}
    \item An L1 Blockchain (e.g., Ethereum): It hosts the main ORU smart contract $\contract$. This contract handles state root submissions, dispute resolution logic, staking, and the RAT protocol execution.
    \item Optimistic rollup (L2): It processes transactions off-chain and periodically posts state commitments and transaction data (for data availability) to the L1 blockchain.
    \item Proposer (P): A single, designated entity responsible for ordering L2 transactions, computing the resulting state transitions, and submitting the corresponding state commitments to the L1 contract.
    \item Registered Validators ($\validators$): A set of $N$ registered validators $\validators = \{\validator{1}, \validator{2}, ..., \validator{N}\}$, each identified by an L1 address and required to stake enough collateral managed by the L1 contract. Instead of receiving a continuous validator fee, these validators are subject to a duty-of-attention test. For permissionless security, not only this registered set but also any validator can submit a fraud proof against a posted state root.
\end{itemize}

In addition to the system model, the protocol relies on the following standard security and operational assumptions:

\begin{itemize}
    \item \textbf{L1 Security:} The underlying L1 blockchain provides security guarantees, including finality for confirmed transactions and resistance to censorship.
    \item \textbf{Cryptographic Primitives:} The cryptographic hash function $H(\cdot)$ (e.g., Keccak-256) used is collision-resistant and behaves like a random oracle for unpredictability purposes.
    \item \textbf{Network Model (Partial Synchrony):} The network ensures that messages broadcast by honest parties are eventually delivered to other honest parties within a finite, bounded time, although the exact bound may not be known a priori or may fluctuate. 
    \item \textbf{Homogeneous Validators and Linear Costs:} The validators are homogeneous in terms of risk preference and face a common, linear cost for operation. This cost largely reflects metered, usage-based computational expenses (e.g., cloud infrastructure) proportional to their verification activities. \label{assumption:validators}
    \item \textbf{Dispute Game Security:} The dispute game of ORU systems are safe and live. If a validator disputes an incorrect state commitment by a proposer, the validator wins the dispute game. If a validator disputes a correct state commitment by a proposer, the proposer wins the dispute game.
\end{itemize}

\begin{figure}[htb]
    \centering
    \includegraphics[width=0.7\linewidth]{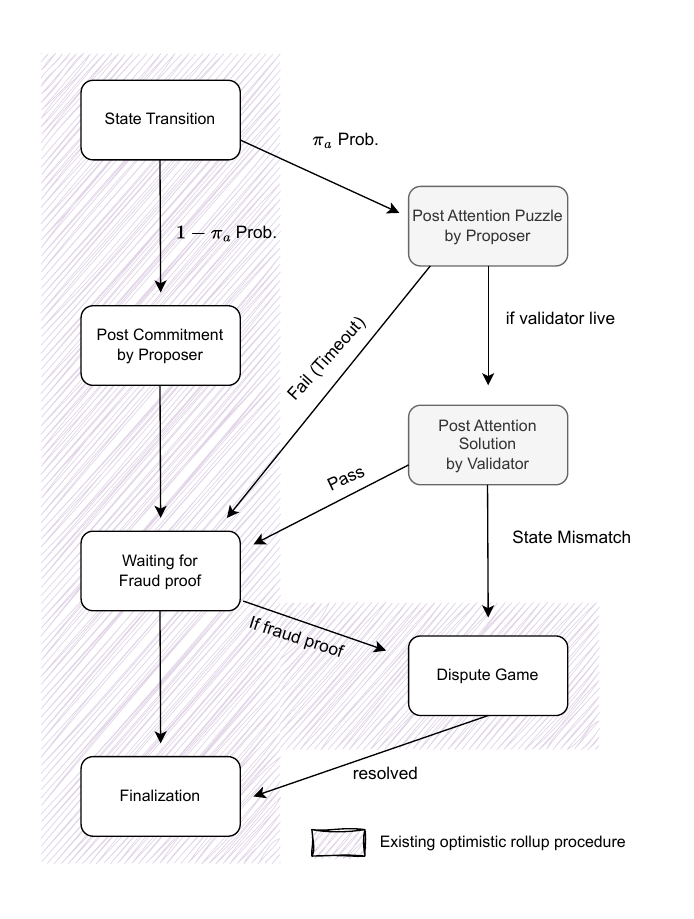}
    \caption{Overview of RAT Integrated with Optimistic Rollup Procedure}
    \label{fig:overview}
\end{figure}

\subsection{Protocol Flow and Challenge Trigger}
\label{subsec:overall_flow}

The RAT protocol introduces a probabilistic verification step into the standard optimistic rollup procedure (Commitment $\rightarrow$ Dispute window $\rightarrow$ Finalization) for state commitments. It aims to ensure validator attentiveness without significantly altering the core ORU dispute window. 

Figure~\ref{fig:overview} illustrates this integrated flow. Deviating from the standard Optimistic Rollup procedure, following an L2 state transition, the proposer interacts with the L1 smart contract. With probability $1 - \pi_a$, the proposer posts its state commitment directly. Conversely, with probability $\pi_a$, the proposer initiates an attention test by submitting an ``attention puzzle'' to the L1 smart contract. The design of the attention puzzle requires to two key properties:

\begin{itemize}
    \item \textbf{Commitment Binding:} The attention puzzle must be uniquely and verifiably bound to the specific state commitment ($\sigma_P$) being attested in that epoch. This ensures the test is relevant to the current state transition. This property is key to ensuring that the ORU's dispute window does not need to be extended for the attention test.
    \item \textbf{Knowledge Proof:} Only a validator possessing the correct L2 state transition information should be able to derive and submit the valid solution to the puzzle.
\end{itemize}

If an attention test is triggered, a target validator is selected by $\contract$ from the registered validator pool. Based on the target validator's response (or lack thereof) within a predefined time window, there are three possible outcomes:

\begin{itemize}
    \item \textbf{\textsc{Pass}:} The validator submits a solution, which the contract verifies against the proposer’s attention test. Then, the standard ORU procedure continues, awaiting potential fraud proofs during the remainder of the dispute window.
    
    \item \textbf{\textsc{Timeout} (Fail):} The validator does not submit any solution within the attention test timing window. The validator gets a specific penalty to the deposit. Then, the standard ORU procedure continues, awaiting potential fraud proofs during the remainder of the dispute window.

    \item \textbf{State Mismatch:} If the state transition as computed by the validator does not match the L2 state ($\sigma_P$) posted by the proposer, a dispute challenge against the posted transition must be initiated.\footnote{We intentionally do not specify the challenger of the dispute. This is because another validator---not the one targeted by the attention test---may initiate the dispute before the designated validator reacts. Once the dispute game is initiated, it serves as sufficient evidence of a state mismatch.}

\end{itemize}

The subsequent subsection will detail the specific mechanisms for implementing this RAT protocol flow and ensuring the aforementioned properties.

\subsection{Attention Test Mechanism}
\label{sec:attention_test_mechanism}

Our mechanism is designed to be integrated into existing ORU protocols with minimal manipulation to their core logic. A key aspect of this efficiency is leveraging the proposer's standard operational flow. In typical ORUs, the proposer periodically submits a state commitment with the following L2 block number to the L1 smart contract ($\mathcal{SC}$). This commitment is usually the state root ($\sigma_P$) of a Merkle tree that compactly represents the entire L2 state. Being a root hash, $\sigma_P$ itself does not reveal the underlying L2 data but attests to it.

The RAT mechanism utilizes this existing state commitment process. The proposer always submits its computed state root $\sigma_P$ along with associated metadata like the corresponding L2 block number (\texttt{L2blocknum}) to $\mathcal{SC}$, irrespective of whether an attention test is triggered for that epoch. Then, $\mathcal{SC}$ determines if an attention test should be triggered for this state transition. If a test is triggered, the smart contract determines which validator to be tested. Both decisions rely on the current block's timestamp and the previous block's hash for randomization\footnote{L1 blockhash and timestamp are susceptible to manipulation by a L1 block builder. However, as Observation 2 in Section \ref{subsec:robustness_exploited_randomness}, manipulating this randomness to avoid triggering attention tests for fraud does not affect the intended security of the protocol.}.

If an attention test is triggered, the ``attention puzzle'' posed to the target validator is not to re-submit the entire state root (as the proposer has already submitted $\sigma_P$). Instead, to prove attentiveness more granularly yet efficiently, the target validator is required to submit the two direct child components, left and right child nodes/hashes, denoted $(L_V, R_V)$, that constitute the state root $\sigma_P$. The two values can be derived by the transaction execution. That is, if $H_{\text{tree}}(L_V, R_V) = \sigma_P$, then the solution is valid. Only a validator that has correctly processed the L2 state transitions leading to $\sigma_P$ would know these immediate constituent components. This two-component solution forms the core of our RAT interaction, depicted in Figure~\ref{fig:interaction}. \\

\textbf{Phase 1: Proposer's State Commitment and Potential Attention Test Trigger by $\mathcal{SC}$}
The Proposer ($P$) computes the L2 state root $\sigma_P$ and the associated \texttt{L2blocknum} for the current epoch. $P$ then sends an L1 transaction to $\mathcal{SC}$ containing $\sigma_P$ and \texttt{L2blocknum}. This submission is identical to the standard ORU procedure where $P$ posts its state commitment.

Upon receiving this transaction, $\mathcal{SC}$ first records $\sigma_P$ and \texttt{L2blocknum}. Then, $\mathcal{SC}$ uses $\sigma_P$ and a source of on-chain data to:
\begin{enumerate}
    \item Probabilistically decide if an attention test is triggered with target probability $\pi_a$. 
    \item If triggered, probabilistically select a target validator $\validator{i}$ from the registered validator set. 
\end{enumerate}
If an attention test is triggered for a certain validator, $\mathcal{SC}$ emits an attention event. This event contains the Proposer's submitted state root $\sigma_P$, \texttt{L2blocknum}, and the target validator's address. This event signals the validator to respond and initiates a response timer, $T_{\text{response}}$. If no test is triggered, the process continues as a standard ORU state commitment. \\

\textbf{Phase 2: Target Validator's Attention Solution Submission}
Upon recognizing the attention event specifically targeting it, the target validator should act. Its core task, which it should have already performed if diligent, is to independently process all L2 transactions up to \texttt{L2blocknum} to derive the state root and, crucially for this test, the two direct child components $(L_{V}, R_{V})$ that hash to form the Proposer's submitted $\sigma_P$.

Within the $T_{\text{response}}$ window, the validator must submit its computed $(L_{V}, R_{V})$ as the ``attention solution'' to $\mathcal{SC}$ via an L1 transaction. Upon receiving this solution, $\mathcal{SC}$ performs the verification: it computes $H(L_{V}, R_{V})$ using the L2-specific tree hashing function and compares the result against the $\sigma_P$ previously submitted by the proposer and recorded by $\mathcal{SC}$ for this test. This comparison, along with timeout handling, determines the outcome of the Attention Test (Pass, Fail-Timeout, or State Mismatch).

In this structure, we can expect that the L1 gas cost will not be significant for the successful attention test outcome (Pass). It requires only on additional L1 transaction with two hash strings, one hash calculation, and one comparison of two hash strings.

\begin{figure}[htb]
    \centering
    \includegraphics[width=0.9\linewidth]{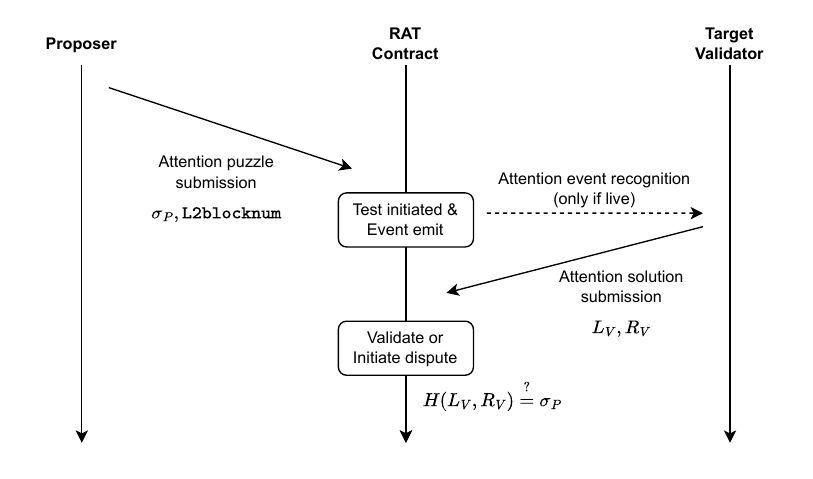}
    \caption{Interaction flow for an attention test}
    \label{fig:interaction}
\end{figure}

%% file: 5.game_model.tex
\section{Game Model}
\label{sec:game_model}

In this section, we develop a game-theoretic model, a game $\mathcal{G}$, before analyzing the strategic interactions by the RAT protocol.

\subsection{Players and Strategies}
\label{sec:players_strategies}

The game $\mathcal{G}$ involves two strategic player types:
\begin{itemize}
    \item \textbf{Proposer (P):} This player is responsible for submitting L2 state commitments to the L1 contract. The proposer chooses a probability $\pi_p \in [0,1]$ of submitting an honest state root (action \textsf{H}). Consequently, the probability of submitting a fraudulent state root (action \textsf{F}) is $1-\pi_p$.
    \item \textbf{Representative Validator (V):} This player represents one of the $N$ homogeneous validators. Each validator $i$ chooses a probability $\pi_{v,i} \in [0,1]$ of being online and actively verifying state transitions during a given period. The probability of being offline (not verifying) is $1-\pi_{v,i}$. Due to the homogeneity assumption in Section \ref{assumption:validators}, we will often focus on a symmetric strategy profile where $\pi_{v,i} = \pi_v$ for all $i$.
\end{itemize}

\subsection{Key Parameters}
\label{sec:parameters_assumptions}

The economic environment of game $\mathcal{G}$ is defined by the following parameters and understanding. A parameter written in uppercase signifies a large value.

\begin{itemize}
    \item \textbf{Validator's Payoffs and Costs:}
    \begin{itemize}
        \item $f_v$: Standard fee received by a validator for its validation activities unless the system fails to detect a fraud state commitment.
        \item $c_m$: Marginal cost incurred by a validator \textit{only when it is online} and actively verifying.
        \item $c_v$: Operating cost incurred by a validator who chose $\pi_v$, and $c_v = \pi_v c_m$.
        \item $r_v$: Total reward pool available if a fraudulent state root submitted by the proposer is successfully challenged.
        \item $r_v^*$: The expected share of $r_v$ for an individual attentive validator that successfully challenges a fraud. We define $r_v^* = \frac{r_v}{1 + (N-1)\bar{\pi}_v}$, where $\bar{\pi}_v$ is the average probability that any other given validator in the system is online. In a symmetric equilibrium, $\bar{\pi}_v = \pi_v$.
        \item $c_{\text{off}}$: Penalty incurred by a validator if it is selected for an Attention Test and fails to respond correctly (i.e., is offline or unresponsive).
        \item $C_{\text{fail}}$: Penalty incurred by each validator if a proposer's fraud goes undetected and leads to a system failure.
    \end{itemize}
    \item \textbf{Proposer's Payoffs and Costs:}
    \begin{itemize}
        \item $f_p$: Standard fee received by the proposer for submitting an honest state root.
        \item $C_{\text{fraud}}$: Penalty incurred by the proposer if its fraudulent submission is detected and successfully challenged.
        \item $R_{\text{fraud}}$: Illicit profit gained by the proposer if its fraudulent submission goes undetected and is finalized.
    \end{itemize}
    \item \textbf{System Parameters:}
    \begin{itemize}
        \item $N$: The total number of registered validators in the system.
        \item $\pi_a$: The intended probability that an Attention Test is triggered for a given state submission. If triggered, one validator is chosen uniformly at random to be tested, so an individual validator is tested with probability $\pi_a/N$.
    \end{itemize}

\end{itemize}

\subsection{Payoff Structure and Utility Calculation Preliminaries}

\label{sec:payoff_preliminaries}
The payoffs for each player depend on their own actions, the actions of the other player(s), and the outcomes of probabilistic events like attention tests and fraud detection. We will use a detailed state-dependent payoff structure as Table~\ref{tab:payoff_ptructure} to derive the expected utilities. The core interactions determining payoffs include:

\begin{itemize}

    \item \textbf{Fraud Detection:} An online validator is assumed to always detect a fraudulent submission (\textsf{F}) by the proposer. An offline validator never detects fraud.

    \item \textbf{Attention Test Outcome:} If an attention test is triggered ($\pi_a$) and a specific validator is challenged (probability $1/N$ for that validator):

    \begin{itemize}
        \item If the challenged validator is online, it passes; no direct reward\footnote{We can consider $f_v$ as an indirect reward for validation.} or penalty from the test itself beyond normal operation.
        \item If the challenged validator is offline, it incurs penalty $c_{\text{off}}$ by timeout.
    \end{itemize}

    \item \textbf{Consequences of Proposer's Fraud:}
    \begin{itemize}
        \item If fraud is detected (e.g., by at least one online validator): The proposer incurs $C_{\text{fraud}}$. The successful challenger(s) share $r_v$ (leading to $r_v^*$ for an individual).
        \item If fraud is not detected: The proposer gains $R_{\text{fraud}}$. All validators (or at least those implicated by the system failure rules) may incur $C_{\text{fail}}$.
    \end{itemize}
\end{itemize}

Based on these interactions and the payoff table, we can calculate the expected utility for each player.

\begin{table}[tb]
\centering
\caption{State-Dependent Payoff Structure for Validator and Proposer}
\label{tab:payoff_ptructure}
\begin{adjustbox}{width=\textwidth, center} %
\begin{tabular}{ccccc}
\toprule
\multicolumn{3}{c}{\textbf{Validator Network State}} & \multicolumn{2}{c}{\textbf{Proposer Action}} \\
\cmidrule(lr){1-3} \cmidrule(lr){4-5}
Network Status & Validator Status & Attention Test & Honest (\textsf{H}) & Fraud (\textsf{F}) \\
\midrule
\multirow{4}{*}{Online} & \multirow{2}{*}{Online (\textsf{On})} & \textsc{Pass} & ($f_v-c_v$, $f_p$) & ($f_v-c_v+r^{*}_v$, $-C_{\text{fraud}}$) \\
& & \textsc{No test} & ($f_v-c_v$, $f_p$) & ($f_v-c_v+r^{*}_v$, $-C_{\text{fraud}}$) \\
\cmidrule(lr){2-5}
& \multirow{2}{*}{Offline (\textsf{Off})} & \textsc{Timeout} & ($f_v-c_v-c_{\text{off}}$, $f_p$) & ($f_v-c_v-c_{\text{off}}$, $-C_{\text{fraud}}$) \\
& & \textsc{No test} & ($f_v-c_v$, $f_p$) & ($f_v-c_v$, $-C_{\text{fraud}}$) \\
\midrule
\multirow{2}{*}{Offline} & \multirow{2}{*}{Offline (\textsf{Off})} & \textsc{Timeout} & ($f_v-c_v-c_{\text{off}}$, $f_p$) & ($-c_{\text{off}}-C_{\text{fail}}$, $f_p+R_{\text{fraud}}$) \\
& & \textsc{No test} & ($f_v-c_v$, $f_p$) & ($-C_{\text{fail}}$, $f_p+R_{\text{fraud}}$) \\
\bottomrule
\end{tabular}
\end{adjustbox}
\end{table}

%% file: 6.equilibrium.tex
\section{Equilibrium Analysis of the RAT Mechanism}
\label{sec:equilibrium_analysis}

This section analyzes the strategic interactions between the proposer and validators within the RAT-enhanced Optimistic Rollup framework. We first derive the expected utility functions for each player and then identify the conditions under which an Ideal Security Equilibrium -- where all validators are attentive and the proposer is honest -- can be achieved and maintained. We also examine the robustness of this equilibrium to potential exploitation of the weak randomness in RAT.

\subsection{Expected Utility Functions}
\label{subsec:expected_utility_functions}

Firstly, Proposition \ref{prop:general_symmetry} and Corollary~\ref{cor:rat_symmetry} allows us to focus on symmetric Nash Equilibria, where $\pi_{v,i} = \pi_v$ for all validators $i$, and thus the expected online probability of other validators $\bar{\pi}_v$ can be set to $\pi_v$.

\begin{proposition}[Symmetry in Games with Monotonic Incentives]
\label{prop:general_symmetry}
In a game with a set of homogeneous players, if the net incentive for any player to take a specific action is a strictly monotonic function of the aggregate number of other players taking that action, then any Nash Equilibrium of the game must be symmetric.
\end{proposition}

\begin{proof}
Let the game consist of a set $V$ of homogeneous players. Let $\Delta(k)$ be a player's net incentive to take an action, where $k$ is the number of other players taking that action. Assume $\Delta(k)$ is strictly monotonic (e.g., strictly decreasing).

Assume, for the sake of contradiction, that an asymmetric Nash Equilibrium strategy profile $\boldsymbol{\pi}^*$ exists. This implies there are at least two players, $i$ and $j$, with different strategies $\pi^*_{i} > \pi^*_{j}$. Let $\pi_k$ be the probability of player $k$ taking the action.

The best-response conditions require that if $\pi^*_{k} > 0$, then the incentive must be non-negative, and if $\pi^*_{k} < 1$, the incentive must be non-positive. Our assumption $\pi^*_{i} > \pi^*_{j}$ thus implies:
1. The incentive for player $i$, $\Delta(\pi^*_{v,-i})$, must be $\geq 0$.
2. The incentive for player $j$, $\Delta(\pi^*_{v,-j})$, must be $\leq 0$.

However, since players are homogeneous and $\pi^*_{i} > \pi^*_{j}$, the aggregate action probability faced by player $j$ is strictly greater than that faced by $i$. Due to the strict monotonicity of $\Delta$, this implies:
3. $\Delta(\pi^*_{v,-j}) < \Delta(\pi^*_{v,-i})$.

Combining these conditions yields the contradiction $\Delta(\pi^*_{v,-j}) \leq 0 \leq \Delta(\pi^*_{v,-i})$ and the strict inequality $\Delta(\pi^*_{v,-j}) < \Delta(\pi^*_{v,-i})$. This is a logical impossibility. Thus, no asymmetric equilibrium can exist.
\end{proof}

\begin{corollary}[Symmetry in the RAT Game]
\label{cor:rat_symmetry}
The RAT game $G$, played by homogeneous validators, satisfies the conditions of Proposition~\ref{prop:general_symmetry}. Therefore, any Nash Equilibrium of $G$ must be symmetric.
\end{corollary}

\begin{proof}
The players (validators) in game $G$ are defined as homogeneous. The net incentive for a validator to be \textsf{Online}, denoted $\Delta$, is a function of the aggregate \textsf{Online} probability of other validators. As established in the main text, this incentive is strictly decreasing because a higher participation rate from others diminishes an individual's expected fraud-detection reward while also reducing the risk associated with being \textsf{Offline}. Since the premises of Proposition~\ref{prop:general_symmetry} hold, the conclusion follows directly.
\end{proof}

\subsubsection{Validator's Expected Utility}
\label{subsubsec:validator_utility}

The expected utility for a representative validator depends on their own strategy (Online or Offline), the proposer's strategy (Honest or Fraudulent, with probability $\pi_p$ of being honest), and the aggregate behavior of other validators (with probability $\pi_v$ of being online, due to Proposition~\ref{prop:general_symmetry}).

\begin{itemize}
    \item \textbf{If the validator chooses to be Online (\textsf{On}):} They incur cost $c_v = c_m \pi_v = c_m \cdot 1 = c_m$. If the proposer is fraudulent (probability $1-\pi_p$), they receive an expected reward share $r_v^* = r_v / (1 + (N-1)\pi_v)$ for detecting fraud.
    \begin{equation}
        U_v(\textsf{On} | \pi_p, \pi_v) = f_v - c_m + (1 - \pi_p) r_v^*
        \label{eq:Uv_Online}
    \end{equation}

    \item \textbf{If the validator chooses to be Offline (\textsf{Off}):} They do not incur any operational cost ($c_v = c_m \pi_v = c_m \cdot 0 = 0$). They face a probability $\pi_a/N$ of being selected for an attention test and incur the slashing penalty $S$ if they fail the test by \textsc{timeout}, regardless of the proposer’s action in this baseline model. If the proposer is fraudulent and the fraud is not detected by any other validator (probability $(1-\pi_v)^{N-1}$)\footnote{This independence-based estimate can differ in real-world deployments if validators’ online/offline choices are correlated (e.g., due to colocation or spot pricing). In that case, $(1-\pi_v)^{N-1}$ should be read as a baseline (best-case) approximation.}, they incur penalty $C_{\text{fail}}$.

    \begin{equation}
        U_v(\textsf{Off} | \pi_p, \pi_v) = \pi_p f_v - \frac{\pi_a}{N} c_{\text{off}} - (1-\pi_p)(1-\pi_v)^{N-1}C_{\text{fail}}
    \end{equation} \label{eq:Uv_Offline_general}
\end{itemize}
The overall expected utility for a validator choosing to be online with probability $\pi_v$ (their own choice) is:
\begin{equation}
E[U_v](\pi_p, \pi_v) = \pi_v U_v(\textsf{On} | \pi_p, \pi_v) + (1 - \pi_v) U_v(\textsf{Off} | \pi_p, \pi_v)
\label{eq:Ev_general}
\end{equation}

\subsubsection{Proposer's Expected Utility}
\label{subsubsec:proposer_utility}

The proposer's expected utility depends on their strategy (Honest or Fraudulent) and the validators' collective online probability ($\pi_v$).

\begin{itemize}
    \item \textbf{If the proposer chooses to be Honest (\textsf{H}):}
    \begin{equation}
        U_p(\textsf{H}) = f_p
        \label{eq:Up_Honest}
    \end{equation}
    \item \textbf{If the proposer chooses to be Fraudulent \textsf{F}:}
    \begin{equation}
        U_p(\textsf{F} | \pi_v) = (1 - (1-\pi_v)^N)(-C_{\text{fraud}}) + (1-\pi_v)^N(f_p + R_{\text{fraud}})
        \label{eq:Up_Fraud}
    \end{equation}
\end{itemize}
The overall expected utility for a proposer choosing to be honest with probability $\pi_p$ is:
\begin{equation}
E[U_p](\pi_p, \pi_v) = \pi_p U_p(\textsf{H}) + (1 - \pi_p) U_p(\textsf{F} | \pi_v)
\label{eq:Ep_general}
\end{equation}

\subsection{Ideal Security Equilibrium}
\label{subsec:ideal_security_equilibrium}

We are interested in an Ideal Security Equilibrium where all validators are online and attentive ($\pi_v = 1$), and the proposer acts honestly ($\pi_p = 1$). For this to be a Nash Equilibrium, neither party should have an incentive to unilaterally deviate.

\subsubsection{Existence Conditions}
\label{subsubsec:existence_conditions_ideal}

We derive the conditions by checking each player's best response under the assumption that the other player is adhering to the ideal strategy.

\paragraph{Condition 1: Validator's Best Response to an Honest Proposer}
Given the proposer is honest ($\pi_p = 1$), a validator will choose to be Online ($\pi_v = 1$) if $U_v(\textsf{On} | \pi_p=1, \pi_v=1) \geq U_v(\textsf{Off} | \pi_p=1, \pi_v=1)$.
(Here, $\pi_v=1$ in the conditioning arguments refers to other validators also being online).

From Eq.~\ref{eq:Uv_Online}, $U_v(\textsf{On} | \pi_p=1, \pi_v=1) = f_v - c_m$.
From Eq.~\ref{eq:Uv_Offline_general}, $U_v(\textsf{Off} | \pi_p=1, \pi_v=1) = f_v - \frac{\pi_a}{N} c_{\text{off}} - (1-1)(1-1)^{N-1}C_{\text{fail}} = f_v - \frac{\pi_a}{N} c_{\text{off}}$.
Thus, the condition is $f_v - c_m \geq f_v - \frac{\pi_a}{N} c_{\text{off}}$, which simplifies to:
\begin{equation}
    c_m \leq \frac{\pi_a}{N} c_{\text{off}} 
    \label{eq:condition_c1}
\end{equation}

\paragraph{Condition 2: Proposer's Best Response to Attentive Validators}
Given all validators are online ($\pi_v = 1$), the proposer will choose to be Honest ($\pi_p = 1$) if $U_p(\textsf{H}) \geq U_p(\textsf{F} | \pi_v=1)$.
From Eq.~\ref{eq:Up_Honest}, $U_p(\textsf{H}) = f_p$.
From Eq.~\ref{eq:Up_Fraud}, $U_p(\textsf{F} | \pi_v=1) = (1 - (1-1)^N)(-C_{\text{fraud}}) + (1-1)^N(f_p + R_{\text{fraud}}) = -C_{\text{fraud}}$.
Thus, we get the below condition which is always true:
\begin{equation}
    f_p \geq -C_{\text{fraud}}
\end{equation}

\begin{theorem}[Existence of Ideal Security Equilibrium]
\label{thm:ideal_security}
The Ideal Security Equilibrium, where all validators are attentive ($\pi_v = 1$) and the proposer is honest ($\pi_p = 1$), exists if and only if the following condition:
\begin{equation}
    c_m \leq \frac{\pi_a}{N} c_{\text{off}}
\end{equation} \label{eq:ideal_security}
\end{theorem} 

\begin{observationbox} \label{observation1}
\textbf{Observation 1.}
The primary condition for Ideal Security reveals a crucial insight: Validator's attentiveness in an honest environment is driven not by direct rewards ($f_v, r_v$), but by the economic trade-off involving operational costs ($c_m$) versus the risk and penalty of failing an Attention Test ($\pi_a, c_{\text{off}}, N$). Consequently, careful calibration of the RAT mechanism's parameters is crucial for guiding the system towards this desired equilibrium.
\end{observationbox}
\vspace{0.5em}

\subsubsection{Stability Analysis for Ideal Security}
\label{subsubsec:stability_ideal}

We now examine whether players have an incentive to deviate from the pure strategies ($\pi_v=1, \pi_p=1$) that constitute the Ideal Security Equilibrium, assuming the conditions from Theorem~\ref{thm:ideal_security} are met.

\paragraph{Validator's Incentive to Maintain $\pi_v=1$ (given $\pi_p=1$)}
The validator's expected utility when the proposer is honest ($\pi_p=1$), as a function of their own online probability $\pi_v$, is 
\begin{equation}
    E[U_v](\pi_v | \pi_p=1) = \pi_v U_v(O | \pi_p=1, \pi_v) + (1-\pi_v) U_v(\textsf{Off} | \pi_p=1, \pi_v).
\end{equation}
Substituting the conditional utilities for $\pi_p=1$, we get
\begin{equation}
    E[U_v](\pi_v | \pi_p=1) = \pi_v (f_v - c_m) + (1-\pi_v) (f_v - \frac{\pi_a}{N} c_{\text{off}}) .
\end{equation}
Taking the derivative with respect to $\pi_v$:
\begin{equation}
    \frac{d E[U_v]}{d \pi_v} \bigg|_{\pi_p=1} = (f_v - c_m) - (f_v - \frac{\pi_a}{N} c_{\text{off}}) = \frac{\pi_a}{N} c_{\text{off}} - c_m
    \label{eq:derivative_Ev}
\end{equation}
If the condition for the ideal security equilibrium in Theorem \ref{thm:ideal_security} holds:
\begin{itemize}
    \item If $c_m < \frac{\pi_a}{N} c_{\text{off}}$, then $\frac{d E[U_v]}{d \pi_v} > 0$. This implies that $\pi_v=1$ is the unique optimal strategy.
    \item If $c_m = \frac{\pi_a}{N} c_{\text{off}}$, then $\frac{d E[U_v]}{d \pi_v} = 0$. The validator is indifferent.
\end{itemize}

\paragraph{Proposer's Incentive to Maintain $\pi_p=1$ (given $\pi_v=1$)}
The proposer's expected utility when all validators are online ($\pi_v=1$), as a function of their own honesty probability $\pi_p$, is $E[U_p](\pi_p | \pi_v=1) = \pi_p f_p + (1-\pi_p) (-C_{\text{fraud}})$.
Taking the derivative with respect to $\pi_p$:
\begin{equation}
    \frac{d E[U_p]}{d \pi_p} \bigg|_{\pi_v=1} = f_p - (-C_{\text{fraud}}) = f_p + C_{\text{fraud}}
    \label{eq:derivative_Ep}
\end{equation}
Given Condition C2 ($f_p \geq -C_{\text{fraud}}$), and standard assumptions $f_p \geq 0, C_{\text{fraud}} > 0$, this derivative is typically strictly positive, making $\pi_p=1$ the unique optimal strategy.

Thus, when the conditions for the Ideal Security Equilibrium are (strictly) met, both players have a clear incentive to adhere to their pure strategies.

\subsection{Robustness to Exploited Weak Randomness}
\label{subsec:robustness_exploited_randomness}

We now consider a more adversarial scenario in which a malicious proposer exploits weaknesses in RAT’s randomness generation. First, we discuss the motivation to manipulate the RAT trigger (to cause or suppress a test) and to bias validator designation. Second, based on this discussion, we show that manipulating trigger suppression under a fraudulent state does not affect the security equilibrium.

\subsubsection{Motivation to Exploit Weak Randomness}
\label{sec:motivation-weak-rand}

We analyze how a malicious proposer might exploit weak randomness along two axes:
(i) whether to cause an attention test (Cause Trigger) or suppress it (Suppress trigger) and
(ii) how to bias the target selection among validators (Targeted designation, Exclusionary designation, or No selection bias).
Table~\ref{tab:weak-rand-motivation} summarizes the resulting $2\times 3$ cases and our assessment of the attacker's motivation in each case. We further assume a non-zero manipulation cost for biasing the trigger/selection probabilities via grinding—e.g. paying a bribe to an L1 block builder.

\paragraph{Deliberately causing a trigger.}
Forcing a test yields no upside for the proposer.
If the attacker targets a colluded validator, the test still cannot be passed under a fraudulent root, so exposure to detection and dispute only increases.
Excluding a specific honest validator only marginally increases the colluder’s selection probability and does not materially lower detection. Moreover, under a fraudulent root the colluded validator cannot pass the test, so expected benefit remains non-positive after grinding cost.
Causing a trigger without any selection bias offers no advantage while further raising detection risk.
Accordingly, the attacker’s incentive to cause a RAT trigger is nonexistent or, at best, \emph{weak}.

\paragraph{Deliberately suppressing a trigger.}
When a test is successfully suppressed, who would have been selected becomes irrelevant.
By contrast, suppressing the test itself directly reduces the chance of a forced on-chain challenge, thus offering a clear incentive to the attacker—subject to the feasibility and cost of grinding.

Across all cases, the only first-order benefit arises from suppressing the trigger RAT. Manipulating designation is either irrelevant or yields negligible gain conditional on a trigger.
This observation motivates the next subsection, where we formalize trigger-suppression and analyze its impact to the security equilibrium.

\newcolumntype{M}{>{\raggedright\arraybackslash}p{0.08\textwidth}} %
\newcolumntype{D}{>{\raggedright\arraybackslash}X}                 %

\begin{table}[t]
\centering
\small
\setlength{\tabcolsep}{6pt}
\begin{tabularx}{\textwidth}{l M D M D M D}
\toprule
 & \multicolumn{2}{c}{\textbf{Targeted designation}} 
 & \multicolumn{2}{c}{\textbf{Exclusionary designation}} 
 & \multicolumn{2}{c}{\textbf{No selection bias}} \\
\cmidrule(lr){2-3}\cmidrule(lr){4-5}\cmidrule(lr){6-7}
& \textbf{Motivation} & \multicolumn{1}{c}{\textbf{Rationale}}
& \textbf{Motivation} & \multicolumn{1}{c}{\textbf{Rationale}}
& \textbf{Motivation} & \multicolumn{1}{c}{\textbf{Rationale}} \\
\midrule
\textbf{\begin{tabular}[c]{@{}l@{}}Cause\\ trigger\end{tabular}} 
& None 
& Colluded validator cannot pass under fraud. 
& Weak 
& Excluding one honest validator neither ensures colluder\footnotemark selection nor reduces detection.
& None 
& Higher detection risk. \\
\midrule
\textbf{\begin{tabular}[c]{@{}l@{}}Suppress\\ trigger\end{tabular}}
& None 
& No test 
& None 
& No test
& Clear 
& No test prevents a forced challenge. \\
\bottomrule
\end{tabularx}
\caption{Motivation and rationale for exploiting weak randomness by a malicious proposer.}
\label{tab:weak-rand-motivation}
\end{table}

\footnotetext{If selected, the colluder still fails under a fraudulent root.}

\subsubsection{Analysis under Proposer's RAT Evasion}
\label{subsubsec:analysis_rat_evasion}

As presented in Table~\ref{tab:modified_payoff_ptructure}, the core modification is that if the proposer is fraudulent ($1-\pi_p$), the effective attention test probability $\pi_a$ becomes 0 for that instance concerning RAT penalties for the validator.

\paragraph{Validator's Expected Utility (re-evaluation for Ideal Equilibrium Check)}
When checking the validator's best response to an honest proposer ($\pi_p=1$), the proposer is, by definition, honest. Thus, the assumption about RAT evasion during fraud does not alter the validator's utility calculation in this specific step.
$U_v(\textsf{On} | \pi_p=1, \pi_v=1) = f_v - c_m$.
$U_v(\textsf{Off} | \pi_p=1, \pi_v=1) = f_v - \frac{\pi_a}{N} c_{\text{off}}$ (since proposer is honest, $\pi_a$ applies).
The condition $c_m \leq \frac{\pi_a}{N} c_{\text{off}}$ remains the same.

\paragraph{Proposer's Expected Utility (re-evaluation for Ideal Equilibrium Check)}
When checking the proposer's best response to attentive validators ($\pi_v=1$), if the proposer considers being fraudulent, they know they can evade RAT. However, attentive validators will still detect the fraud through their own verification, leading to the penalty $-C_{\text{fraud}}$. The evasion of RAT does not change the outcome of fraud detection by diligent validators.
$U_p(\textsf{H}) = f_p$.
$U_p(\textsf{F} | \pi_v=1 \text{, RAT evasion during fraud}) = -C_{\text{fraud}}$.
The condition $f_p \geq -C_{\text{fraud}}$ remains the same.

\begin{corollary}[Robustness of Ideal Security Equilibrium under Weak Randomness]
\label{cor:robustness_ideal_security_evasion}
The conditions for the Ideal Security Equilibrium stated in Theorem~\ref{thm:ideal_security} ($c_m \leq \frac{\pi_a}{N} c_{\text{off}}$ and $f_p \geq -C_{\text{fraud}}$) remain sufficient even if a malicious proposer can fully evade the Randomized Attention Test when submitting a fraudulent state root.
\end{corollary}

\vspace{0.5em}

\begin{observationbox} \label{observation2}
\textbf{Observation 2.}
The randomness source for RAT, potentially derived from L1 block data or proposer inputs, might be exploitable. However, even if a malicious proposer successfully evades an attention test \textit{when submitting a fraudulent state}, this exploitation does not alter the conditions required to achieve the Ideal Security Equilibrium. This is because the validator's incentive to be attentive and the proposer's incentive to be honest remain intact.
\end{observationbox}

\vspace{0.5em}

\begin{table}[tb]
\centering
\caption{State-Dependent Payoff Structure for Validator and Advanced Proposer Model}
\label{tab:modified_payoff_ptructure}
\begin{adjustbox}{width=\textwidth, center} %
\begin{tabular}{ccccc}
\toprule
\multicolumn{3}{c}{\textbf{Validator Network State}} & \multicolumn{2}{c}{\textbf{Proposer Action}} \\
\cmidrule(lr){1-3} \cmidrule(lr){4-5}
Network Status & Validator Status & Attention Test  & Honest (\textsf{H}) & Fraud (\textsf{F}) \\
\midrule
\multirow{4}{*}{Online} & \multirow{2}{*}{Online (\textsf{On})} & \textsc{Pass} & ($f_v-c_v$, $f_p$) & - \\
& & \textsc{No test} & ($f_v-c_v$, $f_p$) & ($f_v-c_v+r^{*}_v$, $-C_{\text{fraud}}$) \\
\cmidrule(lr){2-5}
& \multirow{2}{*}{Offline (\textsf{Off})} & \textsc{Timeout} & ($f_v-c_v-c_{\text{off}}$, $f_p$) & - \\
& & \textsc{No test} & ($f_v-c_v$, $f_p$) & ($f_v-c_v$, $-C_{\text{fraud}}$) \\
\midrule
\multirow{2}{*}{Offline} & \multirow{2}{*}{Offline (\textsf{Off})} & \textsc{Timeout} & ($f_v-c_v-c_{\text{off}}$, $f_p$) & - \\
& & \textsc{No test} & ($f_v-c_v$, $f_p$) & ($-C_{\text{fail}}$, $f_p+R_{\text{fraud}}$) \\
\bottomrule
\end{tabular}
\end{adjustbox}
\end{table}

%% file: 7.design.tex
\section{Design for Ideal Security} \label{sec:design_ideal_security}

A key objective for system designers is to foster the \textbf{Ideal Security Equilibrium}, where all validators are attentive ($\pi_v^*=1$) and the Proposer remains honest ($\pi_p^*=0$). This section discusses how system parameters, particularly the attention test probability ($\pi_a$) and the penalty for failing an attention test ($C_{\text{off}}$), can be tuned to achieve this desirable state while considering practical operational constraints.

\subsection{Consideration of Parameter Interplay}
From Theorem \ref{thm:ideal_security}, the condition for the Ideal Security Pure Strategy Nash Equilibrium is:
\begin{equation} \label{eq:ideal_security_condition_revisited}
c_m \le \frac{\pi_a C_{\text{off}}}{N}
\end{equation}
where $c_m$ is the marginal cost of validator attentiveness, $\pi_a$ is the system-wide probability of an attention test being triggered per epoch, $C_{\text{off}}$ is the penalty for a validator failing such a test (if selected), and $N$ is the number of validators.

In a practical rollup deployment, $c_m$ is largely determined by the rollup's transaction throughput, state complexity, and the operational costs of validator infrastructure. Similarly, $N$ is often an emergent property reflecting the desired level of decentralization, or a target set by the system. Therefore, the primary levers for system designers to satisfy Condition \eqref{eq:ideal_security_condition_revisited} are $\pi_a$ and $C_{\text{off}}$.

A higher $\pi_a$ increases the likelihood of catching inattentive validators but also incurs greater L1 transaction overhead due to more frequent RAT interactions. Conversely, a higher $C_{\text{off}}$ provides a stronger deterrent against inattentiveness, potentially allowing for a lower $\pi_a$. The maximum credible value for $C_{\text{off}}$ is practically bounded by the validator's staked deposit, $D_V$ (i.e., $C_{\text{off}} \le D_V$). To minimize system overhead (by minimizing $\pi_a$), it is rational to set $C_{\text{off}}$ to its maximum credible value, $D_V$.
Rearranging Condition \eqref{eq:ideal_security_condition_revisited} to solve for the minimum required $\pi_a$, assuming $C_{\text{off}} = D_V$, we get:
\begin{equation} \label{eq:min_pi_a_revisited}
\pi_a \ge \frac{c_m N}{D_V}
\end{equation}
This equation forms the basis for our parameter tuning strategy.

\subsection{Estimating Proper Attention Test Frequency} \label{sec:estimate_cm_simulate_pia}
To ground our analysis and understand the practical range for $\pi_a$, we first need to estimate a representative marginal cost of attentiveness ($c_m$) per epoch. The value $c_m$ represents the per-epoch operational cost that a lazy validator might attempt to shirk by not performing its duties (e.g., not re-executing L2 transactions, not maintaining readiness for L1 interactions).

\subsubsection{Estimating Marginal Cost of Attentiveness}
We can approximate $c_m$ by considering the publicly available data regarding the monthly operational costs of running validator nodes for various ORUs. These costs typically encompass server instance specifications (CPU, RAM) and storage requirements. For instance, Table~\ref{tab:validator_costs} outlines the required specifications and estimated operational costs for full nodes of three representative ORUs.

\begin{table}[htbp]
\centering
\caption{Estimated Monthly Operational Costs for Validator Nodes}
\label{tab:validator_costs}
\begin{adjustbox}{width=\textwidth, center} %
\begin{tabular}{l c c c c c c}
\toprule
\multicolumn{1}{c}{Rollup} & \multicolumn{1}{c}{Official Min. Specs} & \multicolumn{1}{c}{\begin{tabular}[c]{@{}c@{}}AWS Instance\\ Example\end{tabular}} & \multicolumn{1}{c}{\begin{tabular}[c]{@{}c@{}}SSD Storage\\ (Min.)\end{tabular}} & \multicolumn{1}{c}{\begin{tabular}[c]{@{}c@{}}Instance\\ Cost/Mo\end{tabular}} & \multicolumn{1}{c}{\begin{tabular}[c]{@{}c@{}}Storage\\ Cost/Mo\end{tabular}} & \multicolumn{1}{c}{\begin{tabular}[c]{@{}c@{}}Total\\ Cost/Mo\end{tabular}} \\ 
\midrule
OP Mainnet \cite{optimism_node} & 16GB RAM, Latest CPU, & r5.xlarge          & 1.6TB & \$140 & \$128 & \$268 \\
           & 1.6TB SSD             & (4vCPU/32GB)       &       &       &       &       \\
\addlinespace
Base \cite{base_node}       & 32GB RAM (64GB rec.), & r5.2xlarge         & 3TB   & \$280 & \$240 & \$520 \\
           & Multi-core, NVMe,     & (8vCPU/64GB)       &       &       &       &       \\
           & 2\textasciitilde3TB+ SSD &                    &       &       &       &       \\
\addlinespace
Arbitrum \cite{arbitrum_node}   & 16GB RAM, 4-core, NVMe, & t3.xlarge          & 1TB   & \$140 & \$80  & \$220 \\
           & (t3.xlarge rec.)      & (4vCPU/16GB)       &       &       &       &       \\
\bottomrule
\end{tabular}
\end{adjustbox}
\end{table}

Based on these figures, it is reasonable to estimate that the monthly operational cost for a validator node performing diligent verification could broadly range from approximately \$200 to 
\$600.
Given this potential range, and adopting a deliberately pessimistic (or perhaps, realistically cautious) stance on operational expenditures for Optimistic Rollups, we set a representative monthly cost of \textbf{\$600 per validator instance} for our analysis. This higher estimate robustly covers the active validation tasks a lazy validator might avoid.

To convert this monthly cost into a per-epoch marginal cost $c_m$, we assume that an L2 state commitment (and thus a potential RAT challenge window or epoch) occurs every hour, given that three representative optimistic rollups---Arbitrum One, Optimism Mainnet, and Base---post their state root approximately once per hour \cite{arbitrum_assertioncreated, optimism_state_root_post, base_state_root_post}. The number of such epochs in a month (approximating a month as 30 days) is:
$$ \frac{30 \text{ days/month} \times 24 \text{ hours/day} }{1 \text{ hour/epoch}} = 720 \text{ epochs/month}. $$
Therefore, the estimated marginal cost of attentiveness per epoch ($c_m$) is:
$$ c_m \approx \frac{\$200 \text{ per month}}{720 \text{ epochs/month}} \approx \$0.833 \text{ per epoch}. $$
This estimated $c_m$ will be used in the subsequent numerical analysis to determine the required attention test probability $\pi_a$.

\subsubsection{Numerical Analysis on Attention Test Probability}

Using the per-epoch marginal operating cost estimated in the previous subsection, \(c_m \approx 0.833\), Figure~\ref{fig:rat_numeric} quantitatively illustrates the relationship among the trigger probability \(\pi_a\), validator set size \(N\), and offline penalty \(C_{\text{off}}\) at the boundary implied by the ideal security equilibrium \(c_m \le \pi_a C_{\text{off}}/N\). It plots the boundary curves \(C_{\text{off}} = \frac{N c_m}{\pi_a}\) in the \((\pi_a, C_{\text{off}})\) plane for \(N \in \{10,25,50,100\}\). The horizontal axis uses a logarithmic scale to cover a wide range of per-epoch trigger probabilities, while the vertical axis shows the offline penalty in linear scale. As \(N\) grows, the curves shift upward, indicating that a larger validator set requires a higher penalty to sustain the same trigger probability; conversely, increasing \(C_{\text{off}}\) reduces the required \(\pi_a\) in inverse proportion.

While increasing \(\pi_a\) raises operating costs (the detailed cost analysis follows in the next section), choosing \(\pi_a\) too small provides insufficient statistical power to verify that RAT is functioning as intended. In practice, we therefore adopt a simple heuristic anchor at \(\pi_a = 1\%\) and check the implied penalty from Figure~\ref{fig:rat_numeric}. Substituting \(\pi_a=0.01\) yields \(C_{\text{off}} = 100\,N\,c_m\) with \(c_m=0.833\), the resulting thresholds are approximately \(\$830\) for \(N=10\), \(\$2080\) for \(N=25\), \(\$4170\) for \(N=50\), and \(\$8330\) for \(N=100\) approximately. In deployment, it is prudent to set \(C_{\text{off}}\) slightly above this 1\% anchor to accommodate estimation error and stochastic variability, while keeping \(\pi_a\) within the operational budget discussed next. 

\begin{figure}[htbp]
    \centering
    \begin{minipage}{0.47\textwidth}
        \centering
        \includegraphics[width=\linewidth]{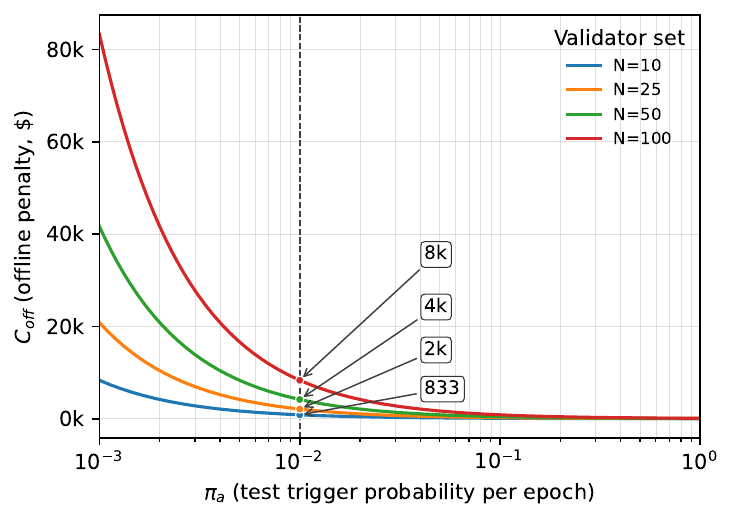}
        \caption{Numerical analysis of the minimum attention test trigger probability required to satisfy the Ideal Security Equilibrium condition (Equation \ref{thm:ideal_security}), across different validator counts and penalty scale.}
        \label{fig:rat_numeric}
    \end{minipage}
    \hfill
    \begin{minipage}{0.512\textwidth}
        \centering
        \includegraphics[width=\linewidth]{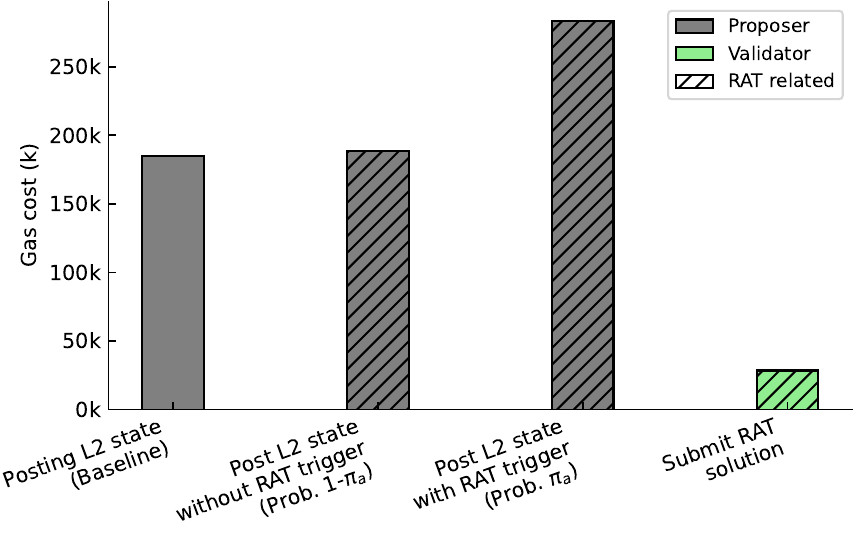}
        \caption{Gas cost comparison with and without the Randomized Attention Test (RAT) deployment in the Optimism Stack. Baseline costs correspond to posting L2 state without the RAT deployment.}
        \label{fig:rat_gas}
    \end{minipage}
\end{figure}

%% file: 8.implementation.tex
\section{Implementation \& Cost Analysis}
\label{sec:implementation}

In this section, we briefly introduce the Proof-of-Concept (PoC) implementation of RAT using Optimism. And we analyze the expected cost overhead by the RAT deployment to optimistic rollups.

\subsection{Implementation}
Our PoC is built on a fork of Optimism Bedrock \cite{ethereum-optimism-optimism}, which is the codebase of representative optimistic rollups including Optimism Mainnet, and compiled with Solidity~0.8.15 (optimizer enabled, EVM version \texttt{cancun}). We use Foundry for testing and gas profiling via \texttt{gasleft()}. See Appendix~\ref{appendix:code} for the PoC code and test scripts. 

Note that in Optimism, \texttt{DisputeGameFactory.create()}  serves as the \emph{entry point that posts the L2 state root} while deploying and initializing a game instance; it is not merely a dispute-only routine as its name might suggest. RAT integrates into this entry point via a conditional branch: \texttt{shouldTriggerRAT()} is evaluated once per creation, and on success \texttt{triggerAttentionTest()} is invoked as a contract call. If the probability check fails, execution returns early with a minimal overhead. If it passes and a valid challenger exists, the path includes randomized selection, bond handling, state writes, and event emission. Since RAT is called internally, the 21{,}000 intrinsic gas applies only once to the outer transaction.

Algorithm~\ref{alg:PoC} presents a pseudo code that illustrates the implemented flow.
On top of the baseline creation path, we derive a randomness seed
from the previous block hash, and timestamp. Both of them provide weak randomness and can be controlled by L1 block builders.
This seed is fed to both the probabilistic trigger and the challenger selection.
After a target validator is selected, we adopt an \emph{upfront slashing} policy: once an attention test is triggered, the validator’s bond is slashed immediately
and refunded only if the test is passed or the validator joins the dispute game.
We did not choose an ex post slashing scheme because it requires an extra party to call for declaring the test failure of the validator and it creates responsibility ambiguity. Our design incurs a small predictable upfront cost
to remove that follow-up call entirely.

\input{RATPoC}

\subsection{Cost Analysis}
Figure~\ref{fig:rat_gas} overviews the gas cost for the transactions in RAT and the baseline transaction which posts a L2 state. For the proposer without RAT, \texttt{create()} uses roughly 184{,}991 gas which is considered baseline. With RAT deployed but not triggered (probability check only), the total rises to about 188{,}695 gas (an overhead of $\approx$3704). When the trigger fires and a valid challenger exists, the total reaches about 283{,}366 gas, i.e., an additional $\approx$98375 over baseline. 
On the validator side, a successful \texttt{submitCorrectEvidence()}, which passes the test with the correct Merkle tree values for the posted state root, call costs about 28520 gas. 

\subsubsection{Operational Overhead}
We set the RAT trigger probability to \(p=0.01\) (i.e., 1\%), as discussed in Section~\ref{sec:design_ideal_security}, an appropriate bond schedule forces a security equilibrium in which rational participants adhere to a \(1\%\) triggering policy. Assuming RAT executes on Ethereum L1, fees are computed under EIP-1559: \(\mathrm{fee}=\mathrm{gasUsed}\times(\mathrm{base\ fee}+\mathrm{priority\ fee})\).

We decompose expected overhead per proposer transaction into (i) the proposer-side increment, arising from conditional RAT execution inside \texttt{create()}, (ii) the validator-side increment, counting exactly one \texttt{submitCorrectEvidence} per trigger under equilibrium, and (iii) the network total as their sum. Using the Ethereum daily average gas series and taking the Jan–Aug 2025 simple mean \(\bar g\approx 2.3\,\mathrm{gwei}\), and the contemporaneous Jan–Aug 2025 average ETH/USD rate \(\overline{\mathrm{ETHUSD}}\approx \$2753\), we obtain:
\[
\begin{aligned}
\mathbb{E}[\Delta \mathrm{gas}_{\text{prop}}]&=(1-p)\cdot3704+p\cdot98375=\mathbf{4650.71},\\
\mathbb{E}[\Delta \mathrm{gas}_{\text{val}}]&=p\cdot28520=\mathbf{285.20},\\
\mathbb{E}[\Delta \mathrm{gas}_{\text{net}}]&=\mathbf{4935.91}.
\end{aligned}
\]
Converting gas to monthly ETH by \(\mathrm{ETH}=\frac{720 \text{ epochs}}{1 \text{ month}} \times \mathrm{gas}\times \bar g\times10^{-9}\) and then to USD via \(\overline{\mathrm{ETHUSD}}\) yields:
\[
\begin{aligned}
\text{Proposer: } & \mathbf{7.70\times10^{-3}}~\mathrm{ETH}/\mathrm{month}\ \ (\approx \mathbf{\$20.88}),\\
\text{Validators: } & \mathbf{4.72\times10^{-4}}~\mathrm{ETH}/\mathrm{month}\ \ (\approx \mathbf{\$1.30}),\\
\text{Network: } & \mathbf{8.21\times10^{-3}}~\mathrm{ETH}/\mathrm{month}\ \ (\approx \mathbf{\$22.18}).
\end{aligned}
\]

While we consider that these figures scale approximately linearly with \(\bar g\) and the ETH/USD rate, the overhead introduced by RAT remains insignificant.

%% file: RATPoC.tex
\begin{algorithm}[t]
\caption{\textsc{CreateGame}: baseline in \textbf{black}, RAT-only lines in \textbf{\textcolor{rat}{blue}}}
\label{alg:PoC}
\KwIn{$gameType$, $stateRoot$, $L2BlockNumber$}
\KwOut{$game$}

$game \leftarrow \mathrm{NewGameInstance}(gameType)$\;
$game.\mathrm{initialize}(stateRoot,\, L2BlockNumber)$\;
$\mathrm{RegisterGame}(game,\, stateRoot,\, gameType)$\;
$\mathrm{EmitGameCreated}(game,\, stateRoot,\, gameType)$\;

\RAT{$prevHash \leftarrow \mathrm{blockhash}(\mathrm{block.number}-1)$}\;
\RAT{$ts \leftarrow \mathrm{block.timestamp}$}\;
\RAT{$seed \leftarrow H(prevHash \Vert ts)$}\;

\RAT{\If{$\mathrm{shouldTriggerRAT}(seed)$}{
  \RAT{$candidate \leftarrow \mathrm{RAT.pickTargetValidator}(seed)$\;}
  \RAT{\If{$candidate.\mathrm{exists}()$}{
    \RAT{$\mathrm{RAT.slashBond}(candidate,\, perTestBondAmount)$ \tcp*{upfront slashing; refunded only if test passes or a legitimate dispute proceeds}}
    \RAT{$\mathrm{RAT.recordAttention}(game:\{ candidate,\, stateRoot,\, bondAmount,\, L1BlockNumber\})$\;}
    \RAT{$\mathrm{RAT.emitTriggered}(stateRoot,\, L2BlockNumber,\, candidate)$\;}
  }}
}}

\textbf{return} $game$\;
\end{algorithm}

%% file: 9.discussion.tex
\section{Discussion}
\label{sec:discussion}

In this section, we discuss the main limitations of our modeling and assumptions.

\subsection{The Simplicity of the Validator Cost Function}

Our analysis utilized a linear cost function for validator attentiveness, a reasonable first-order approximation, particularly given the metered nature of cloud services commonly used by validators. However, the true cost could exhibit non-linearities, such as initial setup costs or diminishing then increasing marginal costs, potentially resembling a sigmoid function. While a comprehensive analysis of such variations is extensive, we posit that as long as the cost function remains monotonically increasing, the fundamental structure of the derived Nash Equilibria (Ideal Security, System Failure, and potential MSNE) would likely persist, with shifts primarily occurring in the specific equilibrium points and the thresholds between them. A more granular understanding of these cost structures could further refine the practical parameter tuning of RAT, although the current model provides a robust foundation for understanding the core incentive dynamics.

\subsection{Limitations of State Transition Tracking in Attention Tests}

The RAT protocol, as designed, effectively verifies a validator's ability to track L2 state transitions and compute correct state roots, which is a critical aspect of attentiveness. Nevertheless, this does not comprehensively guarantee the validator's full capability to successfully generate and submit a valid fraud proof during an actual dispute, as highlighted by incidents like the Kroma Optimistic Rollup's VM misconfiguration issue in 2024 \cite{kroma2024challenge}. This distinction underscores that RAT primarily addresses validator liveness and computational correctness for state tracking, rather than the end-to-end integrity of the fraud proof submission pipeline, which can involve complex L1 contract interactions and specific local environment configurations. While RAT enhances security by ensuring foundational diligence, future research could explore extensions or complementary mechanisms to more holistically assess a validator's readiness to complete the entire fraud proof process, balancing the desire for comprehensive assurance against increased complexity and L1 overhead.

%% file: 10.conclusion.tex
\section{Conclusion}
\label{sec:conclusion}

This paper introduced and detailed the Randomized Attention Test (RAT) protocol, a novel mechanism specifically designed to address the verifier's dilemma concerning validator attentiveness within Optimistic Rollups. By employing probabilistic challenges, RAT proactively monitors and incentivizes validator diligence. Our game-theoretic analysis established the conditions for the Ideal Security Equilibrium, demonstrating that honest and attentive behavior from all participants can be encouraged. Significantly, this desirable state is shown to be practically achievable with modest economic penalties for non-compliance and a low challenge probability, thereby minimizing L1 gas overhead and operational expenses. The analysis also provided a framework for understanding the trade-offs inherent in tuning RAT's parameters. Ultimately, RAT is proposed as a vital complement to existing fraud proof mechanisms, enhancing the security posture of Optimistic Rollup systems by ensuring that validators are not just present, but actively and correctly performing their crucial verification duties. Our implementation data further demonstrate that deploying RAT introduces negligible operational overhead. This work paves the way for more robust and trustworthy optimistic rollup implementations.

%% file: arxiv-main.bbl
\begin{thebibliography}{10}

\bibitem{luu2015demystifying}
L.~Luu, J.~Teutsch, R.~Kulkarni, and P.~Saxena, ``Demystifying incentives in the consensus computer,'' in {\em Proceedings of the 22Nd acm sigsac conference on computer and communications security}, pp.~706--719, 2015.

\bibitem{felten2019attention}
E.~Felten, ``Cheater checking: How attention challenges solve the verifier's dilemma.'' \emph{Medium} (Offchain Labs), \url{https://medium.com/offchainlabs/cheater-checking-how-attention-challenges-solve-the-verifiers-dilemma-681a92d9948e}, Sept. 2019.
\newblock Accessed: 2025-08-31.

\bibitem{park2021optimistic}
A.~Park, ``Optimistic rollup is not secure enough than you think -- game theoretic approach for more verifiable rollup.'' \emph{Medium} (Tokamak Network), \url{https://medium.com/tokamak-network/optimistic-rollup-is-not-secure-enough-than-you-think-cb23e6e6f11c}, May 2021.
\newblock Accessed: 2025-08-31.

\bibitem{mamageishvili2023incentive}
A.~Mamageishvili and E.~W. Felten, ``Incentive schemes for rollup validators,'' in {\em The International Conference on Mathematical Research for Blockchain Economy}, pp.~48--61, Springer, 2023.

\bibitem{tas2022accountable}
E.~N. Tas, J.~Adler, M.~Al-Bassam, I.~Khoffi, D.~Tse, and N.~Vaziri, ``Accountable safety for rollups,'' {\em arXiv preprint arXiv:2210.15017}, 2022.

\bibitem{li2023security}
J.~Li, ``On the security of optimistic blockchain mechanisms,'' {\em Available at SSRN 4499357}, 2023.

\bibitem{Palakkal2024sok}
R.~Palakkal, J.~Gorzny, and M.~Derka, ``Sok: Compression in rollups,'' in {\em 2024 IEEE International Conference on Blockchain and Cryptocurrency (ICBC)}, pp.~712--728, 2024.

\bibitem{mamageishvili2023efficient}
A.~Mamageishvili and E.~W. Felten, ``Efficient rollup batch posting strategy on base layer,'' in {\em International Conference on Financial Cryptography and Data Security}, pp.~355--366, Springer, 2023.

\bibitem{crapis2023eip}
D.~Crapis, E.~W. Felten, and A.~Mamageishvili, ``Eip-4844 economics and rollup strategies,'' in {\em Financial Cryptography and Data Security. FC 2024 International Workshops} (J.~Budurushi, O.~Kulyk, S.~Allen, T.~Diamandis, A.~Klages-Mundt, A.~Bracciali, G.~Goodell, and S.~Matsuo, eds.), (Cham), pp.~135--149, Springer Nature Switzerland, 2025.

\bibitem{Heimbach2025}
L.~Heimbach and J.~Milionis, ``The early days of the ethereum blob fee market and lessons learnt,'' in {\em Proceedings of Financial Cryptography 2025}, 2025.
\newblock Presented at Financial Cryptography 2025, to appear.

\bibitem{huang2024two}
Y.~Huang, S.~Wang, Y.~Huang, and J.~Tang, ``Two sides of the same coin: Large-scale measurements of builder and rollup after eip-4844,'' {\em arXiv preprint arXiv:2411.03892}, 2024.

\bibitem{park2024impact}
S.~Park, B.~Mun, S.~Lee, W.~Jeong, J.~Lee, H.~Eom, and H.~Jang, ``Impact of eip-4844 on ethereum: Consensus security, ethereum usage, rollup transaction dynamics, and blob gas fee markets,'' {\em arXiv preprint arXiv:2405.03183}, 2024.

\bibitem{lee2024180}
S.~Lee, ``180 days after eip-4844: Will blob sharing solve dilemma for small rollups?,'' in {\em Proceedings of the DICG Workshop, in association with the 45th IEEE International Conference on Distributed Computing Systems (ICDCS 2025)}, (Glasgow, UK), 2025.
\newblock to appear.

\bibitem{alvarez2024bold}
M.~M. Alvarez, H.~Arneson, B.~Berger, L.~Bousfield, C.~Buckland, Y.~Edelman, E.~W. Felten, D.~Goldman, R.~Jordan, M.~Kelkar, {\em et~al.}, ``Bold: Fast and cheap dispute resolution,'' {\em arXiv preprint arXiv:2404.10491}, 2024.

\bibitem{nehab2024dave}
D.~Nehab, G.~C. de~Paula, and A.~Teixeira, ``Dave: a decentralized, secure, and lively fraud-proof algorithm,'' {\em arXiv preprint arXiv:2411.05463}, 2024.

\bibitem{berger2025economic}
B.~Berger, E.~W. Felten, A.~Mamageishvili, and B.~Sudakov, ``Economic censorship games in fraud proofs,'' in {\em Proceedings of the 26th ACM Conference on Economics and Computation (EC'25)}, 2025.

\bibitem{lee2025wtsc}
S.~Lee, ``Hollow victory: How malicious proposers exploit validator incentives in optimistic rollup dispute games,'' in {\em Proceedings of the WTSC Workshop, in association with the International Conference on Financial Cryptography and Data Security 2025 (FC 2025)}, (Miyakojima, Japan), 2025.

\bibitem{shengPoD2024AFT}
P.~Sheng, R.~Rana, S.~Bala, H.~Tyagi, and P.~Viswanath, ``{Proof of Diligence: Cryptoeconomic Security for Rollups},'' in {\em 6th Conference on Advances in Financial Technologies (AFT 2024)} (R.~B\"{o}hme and L.~Kiffer, eds.), vol.~316 of {\em Leibniz International Proceedings in Informatics (LIPIcs)}, (Dagstuhl, Germany), pp.~5:1--5:24, Schloss Dagstuhl -- Leibniz-Zentrum f{\"u}r Informatik, 2024.

\bibitem{optimism_node}
{Optimism}, ``Tutorial: Run a node from source.'' \url{https://docs.optimism.io/operators/node-operators/tutorials/run-node-from-source}, 2025.
\newblock Accessed: 2025-05-21.

\bibitem{base_node}
{Base}, ``Base node github repository.'' \url{https://github.com/base/node}, 2025.
\newblock Accessed: 2025-05-21.

\bibitem{arbitrum_node}
{Arbitrum}, ``Run a full arbitrum node.'' \url{https://docs.arbitrum.io/run-arbitrum-node/run-full-node}, 2025.
\newblock Accessed: 2025-05-21.

\bibitem{arbitrum_assertioncreated}
``Arbitrum one {AssertionCreated} event log.'' \url{https://etherscan.io/address/0x4DCeB440657f21083db8aDd07665f8ddBe1DCfc0#events}, 2025.
\newblock Accessed: 2025-09-01.

\bibitem{optimism_state_root_post}
``Optimism mainnet state root post by create.'' \url{https://etherscan.io/address/0xe5965Ab5962eDc7477C8520243A95517CD252fA9}, 2025.
\newblock Accessed: 2025-09-01.

\bibitem{base_state_root_post}
``Base state root post by create.'' \url{https://etherscan.io/address/0x43edB88C4B80fDD2AdFF2412A7BebF9dF42cB40e}, 2025.
\newblock Accessed: 2025-09-01.

\bibitem{ethereum-optimism-optimism}
{ethereum-optimism organization}, ``{Optimism: Ethereum, scaled}.''
\newblock Ethereum Layer-2 scaling solution (OP Stack core components; 21,276 commits; 422+ contributors; primarily Go and Solidity).

\bibitem{kroma2024challenge}
Kroma, ``About the first successful challenge on kroma mainnet,'' April 2024.
\newblock Accessed: 2025-05-25.

\end{thebibliography}
